\documentclass[11pt, preprint]{amsart}

\usepackage{amsmath, latexsym, amsfonts, amssymb, amsthm}
\usepackage{graphicx, color,hyperref,epsfig,caption,wrapfig,subfig}
\usepackage{pstricks}
\usepackage{mathrsfs}
\usepackage{verbatim}

\usepackage{tikz}
\usepackage[left=2cm, right=2cm]{geometry}
\usepackage{float}

\hypersetup{
    linktoc=page,
    linkcolor=red,          
    citecolor=blue,        
    filecolor=blue,      
    urlcolor=cyan,
    colorlinks=true           
}

\numberwithin{equation}{section}

\newtheorem{theorem}{Theorem}[section]
\newtheorem{lemma}[theorem]{Lemma}
\newtheorem{proposition}[theorem]{Proposition}

\newtheorem{remark}[theorem]{Remark}
\newtheorem{definition}[theorem]{Definition}

\theoremstyle{definition}

\renewcommand{\tilde}{\widetilde}          
\DeclareMathSymbol{\leqslant}{\mathalpha}{AMSa}{"36} 
\DeclareMathSymbol{\geqslant}{\mathalpha}{AMSa}{"3E} 
\renewcommand{\leq}{\;\leqslant\;}                   
\renewcommand{\geq}{\;\geqslant\;}                   

\newcommand{\C}{\mathbb{C}}
\newcommand{\R}{\mathbb{R}}
\newcommand{\Z}{\mathbb{Z}}
\newcommand{\N}{\mathbb{N}}

\newcommand{\E}{\mathbb{E}}
\newcommand{\D}{\mathbb{D}}

\usepackage{xparse}

\def\eps{\varepsilon}

\newcommand{\bz}{\textbf{z}}
\newcommand{\bx}{\textbf{x}}
\newcommand{\by}{\textbf{y}}
\newcommand{\cor}{\langle \prod_{i=1}^N V_{\alpha_i}(z_i) \rangle}
\newcommand{\core}{\langle \prod_{i=1}^N V_{\alpha_i,\eps}(z_i) \rangle_\eps}
\newcommand{\ball}[2]{\mathbf{1}_{B^{}_{#2}}}
\newcommand{\ballc}[2]{\mathbf{1}_{(B^{}_{#2})^c}}
\newcommand{\ann}[2]{\mathbf{1}_{A^{}_{#2}}}
\newcommand{\annc}[2]{\mathbf{1}_{(A^{}_{#2})^c}}
\newcommand{\ba}[2]{\mathbf{1}_{A^{}_{#2} \cap B^{}_{#2}}}
\newcommand{\bac}[2]{\mathbf{1}_{A^{}_{#2} \cap (B^{}_{#2})^c}}
\newcommand{\sF}{\mathscr{F}}
\newcommand{\opn}{\operatorname}

\let\ij\ijj
\newcommand{\ij}{{}}
\newcommand{\kj}{{a(j)}}
\newcommand{\tkj}{{b(j)}}
\newcommand{\dx}[1]{\prod_{j=1}^{#1} d^2 x_j}
\newcommand{\dxj}[1]{\prod_{j \in {#1}} d^2 x_\kj d^2 x_\tkj}
\newcommand{\rs}{{\widehat \C}}
\newcommand{\sob}{{H^1(\rs ,g)}}
\newcommand{\sobo}{{H^1_0(\rs ,g)}}
\newcommand{\dsob}{{H^{-1}(\rs ,g)}}
\newcommand{\dsobo}{{H^{-1}_{0}(\rs ,g)}}
\newcommand{\ssob}{{H^s(\rs ,g)}}
\newcommand{\ssobd}{{H^{-s}(\rs ,g)}}
\newcommand{\ssobo}{{H^s_0(\rs ,g)}}
\newcommand{\ssobod}{{H^{-s}_0(\rs , g)}}
\newcommand{\vol}[1]{{\opn{vol}_g(d^2 {#1})}}

\title{Smoothness of correlation functions in Liouville Conformal Field Theory}

\author{Joona Oikarinen}
\address{University of Helsinki, Department of Mathematics and Statistics, P.O. Box 68, FI-00014 University of Helsinki, Finland}
\email{joona.oikarinen@helsinki.fi}
\urladdr{}

\date{\today}
        
\begin{document}

\maketitle

\begin{abstract}
We prove smoothness of the correlation functions in probabilistic Liouville Conformal Field Theory. Our result is a step towards proving that the correlation functions satisfy the higher Ward identities and the higher BPZ equations, predicted by the Conformal Bootstrap approach to Conformal Field Theory.
\end{abstract}

\footnotesize

\normalsize

\tableofcontents

\section{Introduction and main result}
The classical Liouville Field Theory on the Riemann sphere $\rs = \C \cup \{\infty\}$ is a two-dimensional scalar field theory described by the \emph{Liouville Action functional}
\begin{align}\label{intro_action}
S_L(X,g) &= \frac{1}{\pi} \int_\C ( |\partial_z X(z)|^2 + \tfrac{Q}{4}  R_g(z) X(z) g(z) + \pi \mu e^{\gamma X(z)} g(z)) \, d^2z\,.
\end{align}
Here $g(z)|dz|^2$ is some fixed diagonal background metric on the sphere, \\$\partial_z = \frac{1}{2}(\partial_x - i \partial_y)$, $\partial_{\bar z} = \frac{1}{2}(\partial_x + i \partial_y) $ for $z=x+iy$, $R_g$ is the scalar curvature given by $-4 g^{-1} \partial_z \partial_{\bar z} \ln g$ and $\gamma \in (0,2)$, $\mu \in (0,\infty)$ are parameters. In the classical theory one sets $Q_c = \frac{2}{\gamma}$, but we will work with the quantized theory where a renormalization leads to $Q = \frac{2}{\gamma} + \frac{\gamma}{2}$. The two dimensional Lebesgue measure is denoted by $d^2z$. For an account of the classical theory see for example \cite{TZ,TT}.

Quantisation of Liouville theory then amounts to defining the measure $e^{-S_L(X,g)} D X$ on a space of generalized functions from $\rs$ to $\R$ (a negative order Sobolev space) so that the observables $F$ of the random field $X$ are given by the path integral
\begin{align}\label{intro_path}
\langle F \rangle_g &:= \frac{1}{Z} \int F(X) e^{-S_L(X,g)} \, D X\,,
\end{align}
where $Z$ is a normalization constant and $DX$ denotes a formal infinite dimensional Lebesgue measure on the chosen space of generalized functions. The resulting theory exhibits conformal symmetry and is called the Liouville Conformal Field Theory (LCFT). The motivation for studying this theory comes from the hope to understand conformal metrics of the form $e^{\gamma X(z)} g(z) |dz|^2$ on $\rs$, where $X$ is the random field with law (\ref{intro_path}).

The rigorous definition of the path integral (\ref{intro_path}) was given by David-Kupianen-Rhodes-Vargas in \cite{DKRV} by using Gaussian Multiplicative Chaos (GMC) methods. For us the relevant observables of the field will be the correlations of $V_\alpha(z) := e^{\alpha \phi(z)}$, where $\phi(z) = X(z) + \frac{Q}{2} \ln g(z)$. Thus we consider
\begin{align}\label{intro_correlations}
\cor_g  = Z^{-1} \int \prod_{i=1}^N e^{\alpha_i \phi(z_i)} e^{S_L(X,g)} \, D X\,,
\end{align}
where $\alpha_i$ are real numbers (the Liouville momenta) with certain restrictions and $z_i \in \rs$ (the insertions). These are called vertex operators in the physics literature, and they are relevant for understanding the conformal metrics $e^{\gamma X(z)} g(z) |dz|^2$. The correlations (\ref{intro_correlations}) with real weights $\alpha_i$ are relevant for many conjectures related to scaling limits of random planar maps coupled to certain statistical physics models \cite{DKRV, Kup, rv_lectures}. From the CFT point of view one would also want to define these for complex $\alpha_i$. Results in this direction can be found in \cite{dozz} (Theorem 4.1) and \cite{Hua} (Theorem 1.1). In the CFT language the fields $e^{\alpha \phi(z)}$ are supposed to be the primary fields of the LCFT when $\alpha$ belongs to the spectrum, which is supposed to be $Q + i \R$. Many other quantities of the theory are then supposed to be expressible in terms of the correlation functions of the primary fields.

In the physics literature the quantum Liouville theory was first considered in the context of String Theory by Polyakov \cite{Pol} as a building block for Liouville Quantum Gravity. Physics reviews of Liouville theory can be found in \cite{Nak,revisited, Rib}. 

In addition to the path integral formulation, Liouville theory has also been studied (in the physics literature) by using the Conformal Bootstrap method, developed by Belavin-Polyakov-Zamolodchikov in \cite{BPZ}. One goal of the recent mathematical study of LCFT is to unify the path integral and the Conformal Bootstrap approaches, since the equivalence of these two formulations has been controversial even for physicists. Rigorous results in this direction can be found in \cite{ward, dozz, dozz2}. For references on the Bootstrap see e.g. \cite{Rib, revisited}. One mathematical consistency check of the equivalence of the Path Integral and the Conformal Bootstrap was done in \cite{Bav}, where it was shown that the one-point function of LCFT on the torus agrees with the predictions of the Conformal Bootstrap in the large moduli limit. In \cite{BW} the authors derived fusion estimates for the four point function of LCFT on the Riemann sphere and showed that they agree with the predictions of the Conformal Bootstrap.

In this article we will establish the smoothness of the correlation functions (\ref{intro_correlations}) with respect to the insertions $(z_i)_{i=1}^N$, which is required for rigorously proving the CFT structure of Liouville theory, predicted by the Conformal Bootstrap approach (see Section \ref{perspectives}). We now state the exact form of our theorem.

\subsection{Main result}

\begin{theorem}\label{theorem}
Assume that the tuple $(\alpha_i)_{i=1}^N \in \R^N$ satisfies the Seiberg bounds
\begin{align*}
\sum_{i=1}^N \alpha_i > 2Q \,, \quad \alpha_i < Q \quad  \forall i\,,
\end{align*}
and that $g$ is any diagonal Riemannian metric $g(z)|dz|^2$. Then the functions
\begin{align*}
(z_1,\hdots,z_N) \mapsto \cor_g
\end{align*}
are $C^\infty$ on $ U_N := \{ (z_1,\hdots,z_N) \in \C^N : z_i \neq z_j, \forall \, i \neq j\}$.
\end{theorem}

The correlation functions were shown to be $C^2$ on this domain in \cite{ward} and our smoothness proof will be partially based on an iteration of their $C^1$-argument.

\begin{remark}
The generalization of Theorem \ref{theorem} for arbitrary Riemannian metrics is simple since an arbitrary Riemannian metric $g'$ can be written as $g' = \psi^* g$ where $\psi^*$ is a pullback of a diffeomorphism $\psi: \rs \to \rs$ and $g$ is a diagonal metric. The correlation functions are supposed to satisfy diffeomorphism covariance
\begin{align*}
\langle \prod_{i=1}^n V_{\alpha_i}(z_i) \rangle_{g'} &= \langle \prod_{i=1}^n V_{\alpha_i}(\psi(z_i)) \rangle_g\,,
\end{align*}
from which the generalization follows. Of course this requires defining the Liouville theory for arbitrary metrics. On surfaces with genus $2$ or higher this has already been done in \cite{GRV16}.
\end{remark}

\subsection{Perspectives}\label{perspectives}

The smoothness of the correlation functions is needed for the program of deriving the Conformal Bootstrap postulates from the path integral. The Conformal Bootstrap approach predicts that the correlation functions (\ref{intro_correlations}) appear in certain partial differential equations of arbitrarily high order. 

The first set of equations are the \emph{Conformal Ward identities}. These are supposed to emerge from a variation of the background metric $g$. More precisely, let $g = \sum_{i,j=1}^2 g_{ij} dx^i \otimes dx^j$ be a Riemannian metric and fix some nice functions $(f^{ij})_{i,j=1}^2$. We define smooth variations of this metric by $g^{ij}_\eps = g^{ij} + \eps f^{ij}$ where $g^{ij}$ are the components of the inverse matrix of $(g_{ij})_{i,j=1}^2$. Then we expect
\begin{align}\label{se_tensor}
\frac{d}{d\eps} \Big|_{\eps = 0} \cor_{g_\eps} &= \sum_{i,j=1}^2 \frac{1}{4\pi} \int_\C f^{ij}(z) \langle T_{ij}(z) \prod_{k=1}^N V_{\alpha_k}(z_k)  \rangle_g \opn{vol}_g(d^2z)\,, 
\end{align}
where $\opn{vol}_g(d^2z)$ is the volume form of $g$ and $T_{ij}$ is called the \emph{stress-energy tensor}. In a CFT two of the components of $T$ are nontrivial, see \cite{Gaw}. In the $(z,\bar z)$ coordinates they are $T(z) := T_{zz}(z)$ and $\bar T(z) := T_{\bar z \bar z}(z)$. Then according to Belavin-Polyakov-Zamolodchikov \cite{BPZ} the Conformal Ward identities for any non-negative integers $M$ and $N$ are
\begin{align}\label{ward_identities}
\langle T(\zeta) \prod_{i=1}^M T(\zeta_i) \prod_{j=1}^N V_{\alpha_j}(z_j) \rangle_g  &=  \sum_{i=1}^M \left( \frac{2}{(\zeta-\zeta_i)^2} + \frac{1}{\zeta-\zeta_i} \frac{\partial}{\partial \zeta_i} \right) \langle  \prod_{i=1}^M T(\zeta_i) \prod_{j=1}^N V_{\alpha_j}(z_j) \rangle_g \nonumber \\
& \quad + \sum_{j=1}^N \left( \frac{\Delta_{\alpha_i}}{(\zeta-z_j)^2} + \frac{1}{\zeta-z_j} \frac{\partial}{\partial z_j} \right)  \langle  \prod_{i=1}^M T(\zeta_i) \prod_{j=1}^N V_{\alpha_j}(z_j) \rangle_g \nonumber \\
& \quad + \sum_{k=1}^M \frac{c_L}{(\zeta-\zeta_k)^4} \langle \prod_{i=1;i \neq k}^M T(\zeta_i) \prod_{j=1}^N V_{\alpha_j}(z_j) \rangle_g \,. \nonumber
\end{align}
Here $\Delta_\alpha = \frac{\alpha}{2}(Q-\frac{\alpha}{2})$ and $c_L = 1+6Q^2$ is the \emph{central charge} of LCFT. The other nontrivial component $\bar T$ is supposed to satisfy similar identities where each $T$ is swapped for $\bar T$ and the points $\zeta, \zeta_i, z_j$ are swapped for their complex conjugates. In \cite{ward} these identities were proven for $M \in \{0,1\}$. In the proof the authors defined $T(z) = Q \partial^2_z \phi(z) - ((\partial_z \phi(z))^2 + \E[(\partial_z X(z))^2]$ via a regularization procedure and then computed (\ref{ward_identities}) by Gaussian integration by parts. The computations get quite lengthy and thus to prove the identities for all $M$ one should take the variational relation (\ref{se_tensor}) as the definition of $T$.

In the variational computation one should use the fact that two smooth metrics $g'$ and $g$ on $\rs$ are related by 
\begin{align}
g' &= \psi^* (e^\varphi g)\,,
\end{align}
where $\psi: \rs \to \rs$ is a diffeomorphism, $\psi^*$ is the associated pullback and $\varphi: \rs \to \R$ is a smooth function. This means that on the Riemann sphere two metrics are equivalent modulo a diffeomorphism and a conformal factor $e^\varphi$. On higher genus surfaces one has to also take into account the moduli space.

The dependency of $\cor_g$ on $\varphi$ is explicitly given by the Weyl anomaly \cite{DKRV} (Theorem 3.11) and thus the differentiation with respect to this factor is easy. The only thing left to do is to investigate the $\psi$ dependency and for this the smoothness of the correlation functions is needed. In \cite{Gaw} (Lecture 2) this computation is done in a general axiomatic CFT setting where the author assumes the Weyl anomaly, diffeomorphism covariance and some regularity for the correlation functions.

The Ward identities are needed for the construction of representations of the Virasoro algebra. For this the canonical construction of the Hilbert space associated to the LCFT should be carried out and then the generators of the Virasoro algebra should act on some dense subspace of this space, see \cite{Kup}. This will be carried out in a future work.

The other set of partial differential equations that the correlation functions are supposed to satisfy are the \emph{Belavin-Polyakov-Zamolodchikov equations} (BPZ equations). More precisely, the correlation function with the $(r,1)$-degenerate field $\langle V_{- \frac{(r-1) \gamma}{2}} (z) \prod_{i=1}^N V_{\alpha_i}(z_i) \rangle$ is supposed to satisfy the equation (see Section 2 of \cite{BSA})
\begin{align}
\mathcal{D}_r \langle V_{- \frac{(r-1) \gamma}{2}} (z) \prod_{i=1}^N V_{\alpha_i}(z_i) \rangle &= 0\,,
\end{align}
where the differential operator $\mathcal{D}_r$ is given by
\begin{align*}
\mathcal{D}_r &= \sum_{k=1}^r \underset{n_1+\hdots+n_k=r}{ \sum_{(n_1,\hdots,n_k) \in \N^{k}}} \frac{(\frac{\gamma^2}{4})^{r-k}}{\prod_{j=1}^{k-1} (\sum_{i=1}^j n_i)(\sum_{i=j+1}^k n_i)} L_{-n_1} \hdots L_{-n_k}\,,
\end{align*}
and
\begin{align*}
L_{-1} &= \partial_z\,, \\
L_{-n} &= \sum_{i=1}^N \left( - \frac{1}{(z_i-z)^{n-1}} \partial_{z_i} + \frac{\Delta_{\alpha_i} (n-1)}{(z_i-z)^n} \right)\,, \quad n \geq 2\,.
\end{align*}
The degenerate field of order $(1,r)$, given by $\langle V_{- \frac{2(r-1)}{\gamma}} (z) \prod_{i=1}^N V_{\alpha_i}(z_i) \rangle$, is supposed to satisfy a similar equation where $\frac{\gamma}{2}$ gets replaced by $\frac{2}{\gamma}$.

In \cite{ward} the BPZ equations were proven for the $(2,1)$ and $(1,2)$ degenerate fields by using Gaussian integration by parts. The BPZ equations are essential for proving integrability of LCFT. They were used in the proof of the DOZZ-formula \cite{dozz, dozz2} for the $3$-point function of LCFT on the sphere, and after this similar methods were used for obtaining integrability results for one dimensional GMC measures on the unit circle \cite{Remy2} and on the unit interval \cite{Remy3}. The unit circle computation was based on a boundary LCFT, which is defined in \cite{disk}. The connection between the unit interval computation and LCFT is not clear, although the methods used are very similar to the methods used in \cite{dozz}.

\textbf{Acknowledgements.} I would like to thank Antti Kupiainen and Yichao Huang for many fruitful discussions and for giving comments on the manuscript. This project has received funding from the ERC Advanced Grant 741487 (QFPROBA).

\section{Mathematical background}

In this section we quickly review the rigorous definitions behind the probabilistic approach to quantum Liouville theory. Similar discussions can be found in \cite{Kup, Var, DKRV, ward, dozz}.

\subsection{Gaussian Free Field and Gaussian Multiplicative Chaos}

In this paper we will work with the round metric given by
\begin{align}\label{round_metric}
g(z) &= \frac{4}{(1+z \bar z)^2}\,.
\end{align}
For this choice the scalar curvature is constant $R_g(z) =  2$ for all $z \in \rs$.

The usual starting point for defining the measure $e^{-S_L(X)} D X$ is to separate the free field part $\int |\partial_z X|^2  \, dz$ and to think of the measure as
\begin{align}\label{intro_measure}
e^{-  \int ( \frac{1}{4 \pi} Q R_g(z) X(z) + \mu e^{\gamma X(z)}) g(z) \, d^2z } e^{- \frac{1}{\pi}  \int  |\partial_z X(z)|^2 \, d^2z} \, D X\,.
\end{align}
Now the $e^{- \frac{1}{\pi}  \int  |\partial_z X(z)|^2  \, d^2z} \, D X$ factor can be naturally thought of as the (non-normalized) distribution of the Gaussian Free Field (GFF), which formally is a Gaussian process $(X(z))_{z \in \C}$ with covariance
\begin{align}\label{round_covariance}
\E X(z_1) X(z_2) &= C_g(z_1,z_2) :=  \ln \frac{1}{|z_1-z_2|}  - \frac{1}{4} (\ln  g(z_1) + \ln g(z_2)) + \ln 2 - \frac{1}{2}\,,
\end{align} 
and with vanishing mean over the Riemann Sphere
\begin{align}\label{zero_mean}
\int_\C X(z) g(z) \, d^2 z = 0\,.
\end{align}
In other words $C_g$ is the zero mean Green function of the Laplace--Beltrami operator $\Delta_g$.

From (\ref{round_covariance}) we see that the variance $\E X(z_1)^2$ is infinite. This means that in reality the GFF is a random generalized function rather than a function. More precisely, the probability law of $X$ lives on the negative order Sobolev space $\dsob$ which is the continuous dual of $\sob$ (see Appendix \ref{sobolev_appendix} for definitions). Thus $X$ is a Gaussian process $(X(f))_{f \in \sob}$ with covariance given by
\begin{align*}
\E[X(f)X(h)] &= \int_{\C^2} f(x) h(y) C_g(x,y) \, g(x) g(y) d^2 x d^2 y\,, \quad (f,h \in \sob)\,.
\end{align*}
Because of the zero mean property of $X$, we could also think of it as a process $(X(f))_{f \in \sobo}$ indexed by the subspace $\sobo \subset \sob$ of zero mean functions. The fact that $X$ is a random generalized function poses a problem in defining the measure (\ref{intro_measure}) since we would like to think of it as the probability distribution of the GFF, multiplied by some Radon--Nikodym derivative, but now the term $e^{\gamma X(z)}$ becomes ill-defined since the exponential of a generalized function is not defined. This is where the theory of Gaussian Multiplicative Chaos steps in, since it provides a framework for defining exponentials of logarithmically correlated Gaussian fields. This work goes back to Kahane \cite{Kah}. For a more recent review see \cite{review}.

We define $e^{\gamma X(z)} g(z) d^2 z$ to be the $\eps \to 0$ limit of the measures
\begin{align}\label{regularized_chaos}
M_{\gamma,\eps} (d^2z) &:=  e^{\gamma X_\eps(z) - \frac{\gamma^2}{2} \E [X_\eps(z)^2]} g(z) \, d^2 z\,,
\end{align}
where $d^2 z$ is the Lebesgue measure on $\rs$ and $X_\eps(z)$ denotes a regularization which we choose to be a smooth mollification (another common regularization is the circle average). More precisely let $\rho$ be a non-negative $C^\infty(\R)$ function with compact support and define $\rho_\eps(z) = \eps^{-2} \rho(|z|^2/\eps^2)$. Then the regularization of $X$ is defined by $X_\eps = \rho_\eps * X$. We also adapt the notation
\begin{align*}
\frac{1}{(x)_{\eps,\eps}} &:= \rho_\eps * \rho_\eps * \frac{1}{x}\,.
\end{align*}
In \cite{Ber} it was shown that for $\gamma \in (0,2)$ the measures (\ref{regularized_chaos}) converge weakly in probability as $\eps \to 0$ and we denote the limit by $M_\gamma(d^2 z)$. This measure is called the GMC associated to $X$ with respect to the measure $g(z) d^2 z$. The following result goes back to Kahane \cite{Kah} (Lemma 1).

\begin{proposition}\label{kahane} (Kahane Convexity Inequality)
Let $X$ and $Y$ be two continuous Gaussian fields on $\rs$ such that for all $x,y \in \rs$
\begin{align*}
\E[X(x)X(y)] & \leq \E[Y(x)Y(y)]\,.
\end{align*}
Then for all convex $F: \R_+ \to \R$ with at most polynomial growth at infinity and $f: \rs \to \R_+$ we have
\begin{align*}
\E \left[ F \left( \int_\rs f(z) e^{\gamma (X(z) - \frac{\E[X(z)^2]}{2})} \, d^2z  \right)  \right] & \leq \E \left[ F \left( \int_\rs f(z) e^{\gamma (Y(z) - \frac{\E[Y(z)^2]}{2})} \, d^2z  \right)  \right]\,.
\end{align*}
\end{proposition}
When applying Kahane convexity to the GFF one has to use the regularized field $X_\eps$ because of the continuity assumption, but usually this is not a problem.

\subsection{Liouville correlation functions}

We start by defining the path integral (\ref{intro_path}) by setting
\begin{align}\label{rigorous_path}
\langle F \rangle &:= 2 \int_\R e^{-2Qc} \E [ F(X + \tfrac{Q}{2} \ln g(z) + c) e^{- \mu e^{\gamma c } M_\gamma(\C)} ] \, dc\,.
\end{align}
Here $\E$ is the expectation with respect to the GFF $X$ and the integral over $c$ corresponds to a zero mode\footnote{Now $c+X$ is a field where $X$ is the GFF with zero mean and $c$ is distributed according to the Lebesgue measure. Thus the law of $c+X$ is the pushforward of $dc \otimes d \mu_X$ under the map $(c,X) \mapsto c+X$ where $d \mu_X$ is the law of $X$. Note that this is not a finite measure. This field is sometimes called the Massless Free Field.

Another way to view this is to recall that the GFF has the series representation
\begin{align*}
X &= \sqrt{2\pi} \sum_{n=1}^\infty X_n \frac{e_n}{\sqrt{\lambda_n}}\,,
\end{align*}
where $X_n$ are i.i.d. standard Gaussians, and $e_n$ and $\lambda_n$ are the eigenfunctions and eigenvalues of $- \Delta_g$, respectively. Then the zero mode corresponds to adding the constant eigenfunction $e_0$ into this series. This has eigenvalue $0$ so the term in the series would be $\frac{X_0}{\sqrt{0}} e_0$. If we interpret $\frac{X_0}{\sqrt{0}}$ as a Gaussian with infinite variance, that is, the Lebesgue measure, then we end up with the random field $c+X$.}. The observable $F: \dsob \to \R$ is arbitrary as long as the integral converges. This definition is the same as the one given in \cite{dozz}, and differs slightly from the original definition in \cite{DKRV}. The formula comes essentially from plugging in the field $ \phi = c+X + \frac{Q}{2} \ln g$ into a path integral with the regularized action $\int ( \frac{1}{\pi} |\partial_z \phi|^2 + \mu \eps^{\frac{\gamma^2}{2}} e^{\gamma \phi_\eps}) \, d^2 z$ with the Euclidean metric $g \equiv 1$ and adding an integration over the zero mode $c$ and taking the $\eps \to 0$ limit. Absorbing the $g$ dependency to the field via adding the $\frac{Q}{2} \ln g$ term is common, see for example the discussion in Section 3.4 of \cite{rv_lectures}.

Recall that $V_\alpha(z) = e^{\alpha \phi(z)}$, so the vertex operators correspond to the choice $F(X) = e^{\alpha X(z)}$ in (\ref{rigorous_path}). To define $\langle F \rangle$ rigorously for this choice of $F$ one has to regularize the exponential in a similar manner as in the GMC definition above. Thus we define
\begin{align}\label{correlations}
\cor &:= \lim_{\eps \to 0} \langle \prod_{i=1}^N V_{\alpha_i,\eps}(z_i) \rangle_\eps = \lim_{\eps \to 0} 2 \int_\R e^{-2Qc} \E [ \prod_{i=1}^N V_{\alpha_i,\eps}(z_i) e^{- \mu e^{\gamma c } M_{\gamma,\eps} (\C)} ] \, dc\,,
\end{align}
where $\bz = (z_1,\hdots,z_N) \in U_N = \{ (z_1,\hdots,z_N) \in \C^N : z_i \neq z_j, \forall \, i \neq j\}$, and 
\begin{align*}
V_{\alpha,\eps}(z) &= \eps^{\frac{\alpha^2}{2}} e^{\alpha(X_\eps(z) + \frac{Q}{2} \ln g(z) + c) }\,.
\end{align*}
The definition \eqref{correlations} differs from the one given in \cite{DKRV} by a factor of 2 since we decide to match the definition given in \cite{dozz} (which follows the definition in the physics literature).

From now on we denote
\begin{align*}
G_\eps(\bz) &:= \core \,, \quad G_\eps(\bx;\bz) := \langle \prod_{i=1}^n V_{\gamma,\eps} (x_i) \prod_{j=1}^N V_{\alpha_j,\eps}(x_j) \rangle_\eps\,.
\end{align*}
By $G(\bz)$ and $G(\bx;\bz)$ we denote the corresponding $\eps \to 0$ limits. In \cite{DKRV} it was shown that $G(\bz)$ exists and is non-zero if and only if the momenta $\alpha_i$ satisfy the \emph{Seiberg bounds}
\begin{align}\label{seiberg}
\sum_{i=1}^N \alpha_i - 2Q > 0 \,, \quad \alpha_i < Q \quad \forall i \in \{1,\hdots,N\}\,.
\end{align}
In particular this implies that we need $N \geq 3$. By performing the $c$ integral and using the Cameron--Martin theorem we arrive at
\begin{align}\label{reduction_to_chaos}
G(\bz) &= 2 B(\alpha) \mu^{-s} \gamma^{-1} \Gamma(s) \prod_{i<j} \frac{1}{|z_i-z_j|^{\alpha_i \alpha_j}} \E \left[ \left( \int_\C \mathcal{F} (x,\bz) M_\gamma(d^2 x) \right)^{-s}  \right]\,,
\end{align}
where $s = \frac{\sum_{i=1}^N \alpha_i - 2Q}{\gamma}$, $B(\alpha) = 4 e^{- \frac{ \ln 2 - \frac{1}{2}}{2} (s \gamma)^2}$ and
\begin{align*}
\mathcal{F} (x,\bz) &= \prod_{i=1}^N \left( \frac{g(x)^{-\frac{1}{4}}}{|x-z_i|} \right)^{\gamma \alpha_i}\,.
\end{align*}
Thus the correlation functions can be expressed as integrals of explicit functions against the GMC measure. This was initially shown in \cite{DKRV} and by using this formula it is possible to derive fusion estimates that tell us the singular behaviour of $G$ when two of the points $z_i$ merge \cite{ward}. In the smoothness proof we will use a slight generalization of the fusion estimate from \cite{ward} (see Section \ref{section_fusion}). The correlation functions satisfy the following useful integral identity.

\begin{lemma}\label{kpz} (KPZ-identity)
For all $\eps > 0$ and $\bz \in U_N$ we have
\begin{align*}
\mu \gamma \int_\C G_\eps (x;\bz) \,d^2x &= \left(  \sum_{i=1}^N \alpha_i - 2Q \right) G_\eps(\bz)\,.
\end{align*}
The same identity holds when $G_\eps$ is replaced by $G$.
\end{lemma}

\begin{proof}
After the change of variables $c' = c + \gamma^{-1} \ln \mu$ the integral (\ref{correlations}) becomes
\begin{align*}
\core &= 2 \mu^{- \frac{\sum_{i=1}^N \alpha_i-2Q}{\gamma}} \int_\R e^{-2Qc'} \E \left[ \prod_{i=1}^N V_{\alpha_k, \eps} (z_i) e^{- M_{\gamma,\eps}(\C)} \right] \, dc'\,.
\end{align*}
Thus
\begin{align*}
\frac{d}{d\mu} \core &=  \frac{ 2Q - \sum_{i=1}^N \alpha_i }{\mu \gamma} \core\,.
\end{align*}
On the other hand, by differentiating (\ref{correlations}) with respect to $\mu$ we get
\begin{align*}
\frac{d}{d\mu} \core &= - \int_\C G_\eps(x;\bz) \, d^2x\,.
\end{align*}
The identity for the $\eps \to 0$ limits follows from Dominated Convergence, since in \cite{ward} it was shown that $G_\eps$ has an integrable dominant uniformly in $\eps$.
\end{proof}

\subsection{The first derivative of the correlation functions}

Throughout this section we assume that the insertion points are distinct, that is, $\bz \in U_N = \{ (z_1,\hdots,z_N) \in \C^N : z_i \neq z_j, \forall \, i \neq j\}$. In \cite{ward} it was shown by using Gaussian integration by parts that 
\begin{align}\label{derivative}
\partial_{z_i} G_\eps(\bz) &= - \frac{1}{2} \sum_{j=1; j \neq i}^N \frac{\alpha_i \alpha_j}{(z_i-z_j)_{\eps,\eps}} G_\eps(\bz) + \frac{\alpha_i \mu \gamma}{2} \int_\C \frac{G_\eps(x;\bz)}{(z_i-x)_{\eps,\eps}} \, d^2x\,.
\end{align}
A priori we do not know if the $\eps \to 0$ limit of the integral exists. If we just take the $\eps \to 0$ limit of the integrand, the resulting integral does not converge absolutely, thus studying the limit is subtle. In \cite{ward} the convergence was shown by studying the integral transform
\begin{align*}
A_\eps(u) &= \int_\C \frac{G_\eps(x;\bz)}{(u-x)^2 } \, d^2 x := \lim_{\delta \to 0} \int_\C \frac{G_\eps(x;\bz)}{(u-x)^2} \mathbf{1}_{|u-x| > \delta} d^2x \,.
\end{align*}
The limit on the right-hand side exists since the regularized correlation function $G_\eps$ is smooth. By integration by parts we get
\begin{align*}
A_\eps(u) &= - \int_\C \frac{\partial_x G_\eps(x;\bz)}{u-x} \, d^2 x  \\
&= \int_\C \frac{1}{u-x} \Bigg( \sum_{j=1}^N \frac{\alpha_j \gamma}{2} \frac{G_\eps(x;\bz)}{(x-z_j)_{\eps,\eps}} - \frac{\mu \gamma^2}{2} \int_\C \frac{G_\eps(x,y;\bz)}{(x-y)_{\eps,\eps}} \, d^2 y \Bigg) \, d^2 x\,.
\end{align*}
If we contour integrate this relation along the contour $\partial B(z_i,r)$, we get
\begin{align*}
\oint_{\partial B(z_i,r)} A_\eps(u) \, du = 2 \pi i \int_\C \mathbf{1}_{B(z_i,r)}(x) \Bigg( \sum_{j=1}^N \frac{\alpha_j \gamma}{2} \frac{G_\eps(x;\bz)}{(x-z_j)_{\eps,\eps}} - \frac{\mu \gamma^2}{2} \int_\C \frac{G_\eps(x,y;\bz)}{(x-y)_{\eps,\eps}} \, d^2 y \Bigg) \, d^2 x\,,
\end{align*}
where we used the Residue Theorem. On the other hand, the left-hand side can be written as
\begin{align*}
\oint_{\partial B(z_i,r)} A_\eps(u) \, du &= - \oint_{\partial B(z_i,r)} \int_\C \frac{\partial_x G_\eps (x;\bz)}{u-x} \, d^2 x \, du \\
&= - 2 \pi i \int_\C \mathbf{1}_{B(z_i,r)}(x) \partial_x G_\eps(x;\bz) \, d^2 x  
\end{align*}
We end up with the fundamental identity
\begin{align}\label{fundamental_identity}
\int_{B(z_i,r)} \partial_x G_\eps(x;\bz) \, d^2 x  = \int_{B(z_i,r)}   \left( - \sum_{j=1}^N \frac{\alpha_j \gamma}{2} \frac{ G_\eps(x;\bz)}{(x-z_j)_{\eps,\eps}}  + \frac{\mu \gamma^2}{2} \int_\C \frac{G_\eps(x,y;\bz)}{(x-y)_{\eps,\eps}} \, d^2y \right) d^2x \,.
\end{align}
The last term can be written as
\begin{align*}
\int_{\C^2} \mathbf{1}_{B(z_i,r)}(x) \frac{G_\eps(x,y;\bz)}{(x-y)_{\eps,\eps}} \,d^2 y \, d^2x  &= \int_{\C^2} \mathbf{1}_{B(z_i,r)}(x) (\mathbf{1}_{B(z_i,r)}(y) + \mathbf{1}_{B(z_i,r)^c}(y)) \frac{G_\eps(x,y;\bz)}{(x-y)_{\eps,\eps}} \,d^2 y \, d^2x \\
&= \int_{\C^2} F_{1,\eps}(x,y) G_\eps(x,y;\bz) \, d^2 y \, d^2x\,,
\end{align*}
where $F_{1,\eps}(x,y) = \frac{\mathbf{1}_{B(z_i,r)}(x) \mathbf{1}_{B(z_i,r)^c}(y) }{(x-y)_{\eps,\eps}}$. The integral over $\mathbf{1}_{B(z_i,r)}(y)$ vanishes since $\frac{G(x,y;\bz)}{(x-y)_{\eps,\eps}}$ is antisymmetric in $x$ and $y$. The crucial observation is that the $\eps \to 0$ limit of the integrand will turn out to be absolutely integrable. Now (\ref{fundamental_identity}) leads to
\begin{align*}
\frac{\alpha_i \gamma}{2} \int_{B(z_i,r)} \frac{G_\eps(x;\bz)}{(x-z_i)_{\eps,\eps}} \, d^2 x &= - \sum_{j=1;j\neq i}^N \frac{\alpha_j \gamma}{2} \int_{B(z_i,r)} \frac{G_\eps(x;\bz)}{(x-z_j)_{\eps,\eps}} \, d^2x  \\
& \quad + \frac{\mu \gamma^2}{2} \int_{\C^2} F_{1,\eps}(x,y) G_\eps(x,y;\bz) \, d^2y \,d^2 x  \\
& \quad - \oint_{\partial B(z_i,r)} G_\eps(x;\bz) \, d  x\,,
\end{align*}
where the last term comes from integrating by parts the left-hand side of (\ref{fundamental_identity}). From this we can take the $\eps \to 0$ limit and thus we have demonstrated that the $\eps \to 0$ limit of the last term in \eqref{derivative} exists and we see that $\bz \mapsto G(\bz)$ is $C^1$ (checking the $\partial_{\bar z_i}$-derivative works the same way). In the end we obtain
\begin{align}\label{c1}
\lim_{\eps \to 0} \frac{\alpha_i \gamma}{2} \int_\C \frac{G_\eps(x;\bz)}{(x-z_i)_{\eps,\eps}} \, d^2 x &= \int_{\C^2} F(x,y) G(x,y;\bz) \,d^2 x \, d^2 y + \oint_{\partial B(z_i,r)} G(x;\bz) \, d x\,,
\end{align}
where 
\begin{align*}
F(x,y) &= \frac{\mu \gamma^2}{2} F_1(x,y) - \sum_{j=1;j \neq i}^N \left( \frac{\alpha_j \gamma }{2} \frac{\mu}{s+1} \frac{\mathbf{1}_{B(z_i,r)}(x)}{x-z_j} \right) + \frac{\alpha_i \gamma}{2} \frac{\mu}{s+1} \frac{\mathbf{1}_{B(z_i,r)^c}(x)}{x-z_i}\,,
\end{align*}
where $s = \frac{\sum_{j=1}^N \alpha_j - 2Q}{\gamma}$ comes from the KPZ-identity (Lemma \ref{kpz}). The integral
\begin{align*}
\int_{\C^2} F_1(x,y) G(x,y;\bz) \, d^2 x \, d^2 y
\end{align*}
is convergent by the two-point fusion estimate from \cite{ward} (Proposition 5.1) and the following lemma.
\begin{lemma}\label{integral}
Let $K$ be a bounded set containing the origin and $B = B(0,r)$ for some $r>0$. Then
\begin{align*}
\int_K \frac{\mathbf{1}_B(x) \mathbf{1}_{B^c}(y)}{|x-y|^a} \, d^2 x \, d^2 y < \infty
\end{align*}
for $a <3$.
\end{lemma}
\begin{proof}
Denote $x=(x_1,x_2)$ and $y=(y_1,y_2)$. The integral is convergent if and only if the integral
\begin{align*}
\int_0^1 \, dx_1 \int_0^1 \, dx_2 \int_0^1 \, dy_1 \int_1^2 \, dy_2 \, \frac{1}{|x-y|^a}
\end{align*}
is convergent (or we could use squares instead of the balls in our smoothness proof). We can also use the norm $|x-y| = |x_1-y_1| + |x_2-y_2|$ since all norms are equivalent. After a change of variables $z = \varphi(x_1,y_1) = (x_1-y_1,x_1+y_1)$ we get
\begin{align*}
&= \int_{\varphi^{-1}([0,1]^2)} \, d^2z \, \int_0^1 \, d x_2\, \int_1^2 \, d y_2 \, \frac{1}{( C|z_1| + (y_2-x_2))^a} \\
&= \int_{\varphi^{-1}([0,1]^2)} \, d^2z \, \int_0^1 \, d x_2\, \left( \frac{1}{(C|z_1| + (1-x_2))^{a-1}} - \frac{1}{(C|z_1| + (2-x_2))^{a-1}}  \right) \\
&= \int_{\varphi^{-1}([0,1]^2)} \, d^2z \, \Bigg( \frac{1}{(C|z_1|)^{a-2}} - \frac{1}{(C|z_1|+2)^{a-2}} \Bigg)\,.
\end{align*}
The singular part of this integral is around the origin, and computing the part over a small square of the first term yields $a<3$.
\end{proof}

The fusion estimate (Proposition 5.1 from \cite{ward}) tells us that when $x$ and $y$ fuse and the $z_i$'s are away from $x$ and $y$, we have
\begin{align*}
G(x,y;\bz) & \leq C |x-y|^{-2+\zeta}\,,
\end{align*}
where $\zeta > 0$ depends on $\gamma$. For higher derivatives we need a generalization of this estimate for a fusion of $n$ separated pairs of points. 

Finally, the boundary integral term in (\ref{c1}) is convergent since the map $x \mapsto G(x;\bz)$ is continuous when $x$ stays away from the $z_i$'s. The conclusion is that the limit of the integrals on the left-hand side of (\ref{c1}) is expressible as a sum of absolutely convergent integrals. Concretely we have the following
\begin{align}\label{derivative_simplified}
\partial_{z_i} G(\bz) &= - \sum_{j=1;j \neq i}^N \frac{\alpha_j}{2} \frac{1}{z_i-z_j} G(\bz) + \frac{\mu \gamma^2}{2} \int F_1(x_1,x_2) G(x_1,x_2;\bz) \, d^2 x_1 \, d^2 x_2 \nonumber \\
& \quad + \int_\C f(x) G(x;\bz) \, d^2 x +  \oint_{\partial B_1} G(x_1;\bz) \, d x_1\,,
\end{align}
where
\begin{align*}
f(x) &= - \sum_{j=1;j \neq i}^\infty \left( \frac{\alpha_j \gamma }{2} \frac{\mu}{s+1} \frac{\mathbf{1}_{B(z_i,r)}(x)}{x-z_j} \right) + \frac{\alpha_i \gamma}{2} \frac{\mu}{s+1} \frac{\mathbf{1}_{B(z_i,r)^c}(x)}{x-z_i}\,,
\end{align*}
and $s = \frac{\sum_{j=1}^N \alpha_j-2Q}{\gamma}$. Iteration of a process like this is our strategy for the smoothness proof: we will differentiate the resulting absolutely integrable terms in \eqref{derivative_simplified} and the derivatives are given by integrals that do not converge absolutely. Then we simplify these non-absolutely convergent integrals into sums of absolutely convergent integrals and repeat. To properly deal with the non-absolutely convergent integrals one has to replace $G$ with the regularized version $G_\eps$ so that everything is convergent, and then study the $\eps \to 0$ limit.

\section{$n$-pair fusion estimate}\label{section_fusion}

In this section we prove a result concerning the singular behaviour of $G(\bz)$ when multiple pairs of the points $z_i$ merge.

Let $\bz \in U_N = \{ (z_1,\hdots,z_N) \in \C^N : z_i \neq z_j, \forall \, i \neq j\}$ for some $N \geq 3$ and let $n \in \N$. We fix a number $i(j) \in \{1,\hdots,N\}$ for each $j \in \{1,\hdots,n\}$. Define $\delta = \min_{j,k=1;j \neq k}^N\{|z_j-z_k|\}$ and $B_j = B(z_{i(j)},r/j)$ with $r < \delta/2$. Let $\{A_j\}_{j=1}^n$ be disjoint closed annuli containing the circles $\{\partial B_j\}_{j=1}^n$.

\begin{figure}[H]
\begin{tikzpicture}[scale=1.3]

\draw[dashed,thick] (4.5,2.5) circle(55pt);
\draw[thick] (4.5,2.5) circle(50pt);
\draw[dashed,ultra thick] (4.5,2.5) circle(45pt);
\fill[black, fill opacity=0.20] (4.5,2.5) circle(55pt);
\fill[white] (4.5,2.5) circle(45pt);
\draw[thick] (4.5,2.5) circle(25pt);
\draw[dashed,ultra thick] (4.5,2.5) circle(20pt);
\draw[dashed,thick] (4.5,2.5) circle(30pt);
\fill[black, fill opacity=0.2] (4.5,2.5) circle(30pt);
\fill[white] (4.5,2.5) circle(20pt);
\fill[black] (4.5,2.5) circle(1pt);

\draw[dashed,thick] (0,0) circle(41pt);
\draw[dashed,ultra thick] (0,0) circle(29pt);
\fill[black, fill opacity=0.20] (0,0) circle(41pt);
\fill[white] (0,0) circle(29pt);
\draw[thick] (0,0) circle(35pt);
\fill[black] (0,0) circle(1pt);

\fill[black] (1.35,4) circle(1pt);

\fill[black] (1,3) circle(1pt);

\draw (4.65,2.35) node{$z_1$};
\draw (0.15,-0.15) node{$z_2$};
\draw(1.5,3.85) node{$z_3$};
\draw(1.15,2.85) node{$z_4$};

\draw (4.5,1.1) node{$B_1$};
\draw (0, -0.75) node{$B_2$};
\draw (4.5,2.05) node{$B_3$};

\draw (4.5,4.6) node{$A_1$};
\draw (0,1.6) node{$A_2$};
\draw (4.5,3.7) node{$A_3$};
\end{tikzpicture}

\caption{Explanation of the notation in the case $N=4$ and $n=3$ with the choice $(i(1),i(2),i(3)) = (1,2,1)$. The thick lines are the boundaries of the balls $B_i$ and the dashed lines are the boundaries of the annuli $A_i$.}

\label{figure}

\end{figure}
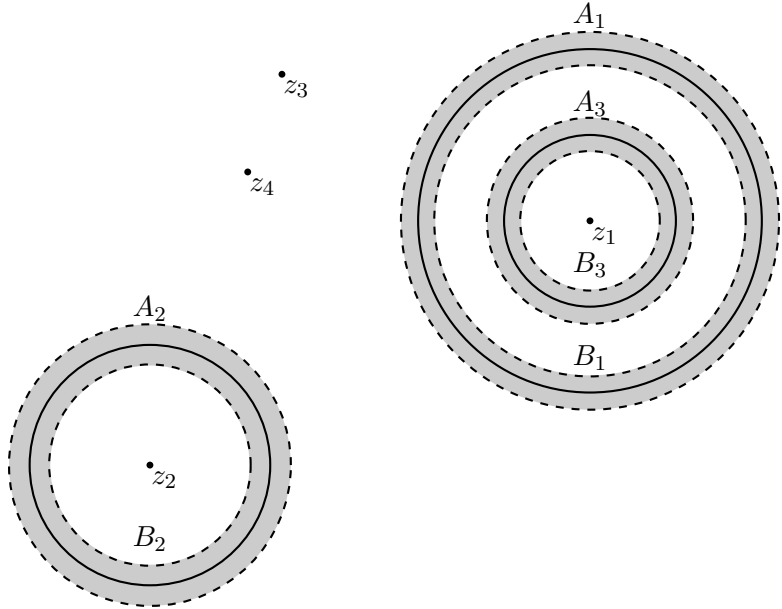

\begin{lemma}\label{fusion}
Denote $\bx = (x_1,\hdots,x_n)$, $\by = (y_1,\hdots,y_n)$ and $\bz = (z_1,\hdots,z_N)$. Let $i(j) \in \{1,\hdots, N\}$ for each $j \in \{1,\hdots, n\}$. Then
\begin{align*}
G(\bx,\by;\bz) \prod_{j=1}^n \ba{i(j)}{j}(x_j) \bac{i(j)}{j}(y_j) & \leq C_{\delta} \prod_{j=1}^n |x_j-y_j|^{-2+\zeta}\,,
\end{align*}
where $\zeta = \zeta(\gamma) > 0$ and $C_\delta$ is a constant depending on $\delta$.
\end{lemma}
\begin{proof}
Our proof will be based on the $n=1$ proof which was done in \cite{ward}. From the proof of Proposition 5.1 in \cite{ward} we get (see Appendix \ref{fusion_appendix} for details)
\begin{align}\label{I_estimate}
G_\eps(\bx,\by;\bz) \prod_{j=1}^n \ba{i(j)}{j}(x_j) \bac{i(j)}{j}(y_j) & \leq C(\delta, \alpha, g) \prod_{j=1}^n |x_j-y_j|^{-\gamma^2} I\,, 
\end{align}
where 
\begin{align*}
I &= \E \left[ \left( \sum_{j=1}^n \int_{D_j}  \frac{1}{ |x-x_j|^{ \gamma^2} |x-y_j|^{\gamma^2}} M_{\gamma}(d^2 x) \right)^{-q} \right]\,,
\end{align*}
and $D_j$ is an annulus around $x_j$ with radii $R_j$ and $|x_j-y_j|$ (we can assume these radii to be so small that the distance between $D_j$ and $\partial A_j$ is positive). The moment $q$ is given by
\begin{align*}
q &= \frac{\sum_i \alpha_i + 2n \gamma - 2Q}{\gamma} = 2n + \frac{\sum_i \alpha_i - 2Q}{\gamma}\,,
\end{align*}
where $\alpha_i$ are assumed to satisfy the Seiberg bounds (\ref{seiberg}). We define the notation
\begin{align*}
W_j &= \int_{D_j}\frac{1}{ |x-x_j|^{ \gamma^2} |x-y_j|^{\gamma^2}} M_{\gamma}(d^2 x)\,, \quad j \in \{1,\hdots,n\}\,.
\end{align*}

\begin{enumerate}
\item \textbf{First assume that the random variables $(W_j)_{j=1}^n$ are independent.} Then we estimate (using first the estimate $(\sum_{j=1}^n W_j)^q \geq \prod_{j=1}^n W_j^{q/n} $ and then independence and finally the proof of Proposition 5.1 in \cite{ward}, summarized in Appendix \ref{fusion_appendix})
\begin{align}\label{ind_estimate}
I & \leq \prod_{j=1}^n \E \frac{1}{W_j^{q/n}} \\
&\leq C \prod_{j=1}^n \sum_{k=0}^\infty (k+1) e^{((2 \gamma -Q) - \frac{q}{n} \gamma)k} \frac{1}{| \ln |x_j-y_j| |^{3/2}} |x_j-y_j|^{ \frac{(2 \gamma - Q)^2}{2}}\,. \nonumber
\end{align}
The series $\sum_k$ is convergent since 
\begin{align*}
2 \gamma - Q - \frac{q}{n} \gamma &= 2 \gamma - Q - 2 \gamma - \frac{\sum_i \alpha_i - 2Q}{n} = - Q - \frac{\sum_i \alpha_i - 2Q}{n} < 0\,.
\end{align*}
The inequality follows from the Seiberg bound (\ref{seiberg}). The claimed estimate follows by simplifying the exponent of $|x_j-y_j|$.

\item \textbf{In reality the random variables $(W_j)_{j=1}^n$ are not independent.} However, we can reduce everything to the independent case by using Proposition \ref{kahane}. Let $X$ be the GFF and $\tilde X = \sum_{j=1}^n \tilde X_j$ be a Gaussian field where $\tilde X_j$ is supported in $D_j$ and has the covariance \eqref{round_covariance}, and all the terms $\tilde X_j$ are independent.  For $x \in D_1$, $y \in D_2$ we have 
\begin{align*}
\E X(x)X(y) &= C_g(x,y) \leq \sup_{x \in D_1, y \in D_2} |C_g(x,y)| =: c_\delta\,.
\end{align*}
Notice also that $\E \tilde X(x) \tilde X(y) = 0$ by definition of $\tilde X$. For $x,y \in D_1$ we have of course $\E X(x) X(y) = \E \tilde X(x) \tilde X(y)$ and thus the inequality
\begin{align*}
\E X(x) X(y) \leq c_\delta + \E \tilde X(x) \tilde X(y)
\end{align*}
holds for all $x,y \in D_1 \cup D_2$. Let $N$ be an independent centered Gaussian with variance $c_{\delta}$. Then 
\begin{align}
\E[X(x)X(y)] &\leq \E [(\tilde X(x)+N)(\tilde X(y)+N)]\,.
\end{align}
Thus by Proposition \ref{kahane}
\begin{align}
\E \left[ \left( \int f(x) M_\gamma(d^2x)  \right)^{-q} \right] & \leq \E \left[ \left( \int f(x) \tilde M_\gamma(d^2x) \right)^{-q} \right] \E [e^{- q \gamma N}] \,,
\end{align}
where 
\begin{align*}
f(x) &= \sum_{j=1}^n \mathbf{1}_{D_j}(x) \frac{1}{|x-x_j|^{\gamma^2} |x-y_j|^{\gamma^2}}\,,
\end{align*}
and $M_\gamma$ is the chaos measure of $X$ and $\tilde M_\gamma$ is the chaos measure of $\tilde X$. Since $N$ is Gaussian, the exponential moment exists and thus the factor $ \E [e^{- q \gamma N}]$ is just a finite constant.  Now we can estimate as in the independent case (\ref{ind_estimate}).
\end{enumerate}
\end{proof}

\section{Sketch of the proof}

In this section we sketch the smoothness proof. The actual detailed proof is given in the next section. Fix some insertion $z_i \in \{z_1,\hdots,z_N\}$ (the index $i$ will now be fixed for the rest of this section) and define $B_j = B(z_i,r/j)$. Note that since $i$ is fixed, all the balls $B_j$ have the same center. By $A_j$ we denote a closed annulus containing $\partial B_j$ such that all the $A_j$'s are disjoint. We define the functions
\begin{align}\label{sketch_F}
F_j(x,y) = \frac{\ball{i}{j}(x) \ballc{i}{j}(y)}{x-y}\,.
\end{align}
Recall the derivative formula
\begin{align*}
\partial_{z_i} G(\bz) &= - \sum_{j=1;j \neq i}^N \frac{\alpha_j}{2} \frac{1}{z_i-z_j} G(\bz) + \frac{\mu \gamma^2}{2} \int F_1(x_1,x_2) G(x_1,x_2;\bz) \, d^2 x_1 \, d^2 x_2  \\
& \quad + \int_\C f(x) G(x;\bz) \, d^2 x +  \oint_{\partial B_1} G(x_1;\bz) \, d x_1\,,
\end{align*}
where
\begin{align*}
f(x) &= - \sum_{j=1;j \neq i}^\infty \frac{\alpha_j \gamma }{2} \frac{\mu}{s+1} \frac{\mathbf{1}_{B(z_i,r)}(x)}{x-z_j} + \frac{\alpha_i \gamma}{2} \frac{\mu}{s+1} \frac{\mathbf{1}_{B(z_i,r)^c}(x)}{x-z_i}\,,
\end{align*}
and $s = \frac{\sum_{j=1}^N \alpha_j-2Q}{\gamma}$.
\subsection{Second derivative}
Next we compute the $i$th partial derivative of the second term in the above formula for $\partial_{z_i} G(\bz)$, which is the most problematic term. By the derivative formula \eqref{derivative} we have
\begin{align*}
\int F_1(x_1,x_2) \partial_{z_i} G(x_1,x_2;\bz) \, d^2 x_1 \, d^2 x_2 &= -\int F_1(x_1,x_2) G(x_1,x_2;\bz) \sum_{j=1;j \neq i}^N \frac{\alpha_i \alpha_j}{2} \frac{1}{z_i-z_j} \, d^2 x_1 \, d^2 x_2 \\
& \quad - \int F_1(x_1,x_2) G(x_1,x_2;\bz) \sum_{j=1}^2 \frac{\alpha_i \gamma}{2} \frac{1}{z_i-x_j} \, d^2 x_1 \, d^2 x_2 \\
& \quad+ \frac{\mu \gamma^2}{2} \int F_1(x_1,x_2) \frac{G(x_1,x_2,x_3;\bz)}{z_i-x_3} \, d^2 x_1 \, d^2 x_2 \, d^2 x_3\,. 
\end{align*}
We have to simplify the terms with $\frac{1}{z_i-x_k}$, $k=1,2,3$, since they are not absolutely convergent. In all the cases the simplification will work in essentially the same way so we focus on the case $k=3$. The simplification follows from an analogue of the identity (\ref{fundamental_identity}):
\begin{align}\label{sketch_f0}
&\int \ball{i}{2}(x_3) F_1(x_1,x_2) \partial_{x_3} G(x_1,x_2,x_3;\bz) \prod_{j=1}^3 d^2 x_j \nonumber \\
& \quad =  -\int \ball{i}{2}(x_3) F_1(x_1,x_2) G(x_1,x_2,x_3;\bz) \sum_{j=1}^N \frac{\alpha_j \gamma}{2} \frac{1}{x_3-z_j}  \prod_{j=1}^3 d^2 x_j \nonumber \\
&\quad  - \int \ball{i}{2}(x_3) F_1(x_1,x_2) G(x_1,x_2,x_3;\bz) \sum_{j=1}^2 \frac{\gamma^2}{2} \frac{1}{x_3-x_j} \prod_{j=1}^3 d^2 x_j \nonumber \\
& \quad + \frac{\mu \gamma^2}{2}  \int \ball{i}{2}(x_3) F_1(x_1,x_2) \frac{G(x_1,x_2,x_3,x_4;\bz)}{x_3-x_4} \prod_{j=1}^4 d^2 x_j\,.
\end{align}
We want to solve for the $j=i$ term (recall that we fixed the index $i$ at the beginning of the section) in the sum $\sum_{j=1}^N$. The other terms in the first sum are automatically convergent. In the last term we insert $1 = \ball{i}{2}(x_4) + \ballc{i}{2}(x_4)$ and it becomes 
\begin{align*}
\frac{\mu \gamma^2}{2}  \int F_2(x_3,x_4) F_1(x_1,x_2) \frac{G(x_1,x_2,x_3,x_4;\bz)}{x_3-x_4} \prod_{j=1}^4 d^2 x_j\,,
\end{align*}
since the integral over $\ball{i}{2}(x_4)$ vanishes by antisymmetry.

Next we deal with the $\frac{1}{x_3-x_j}$, $j=1,2$, terms on the third line of \eqref{sketch_f0}. We proceed similarly as above, that is, we insert $1 = \ball{i}{2}(x_1) + \ballc{i}{2}(x_1)$. We symmetrize the integral over the first factor and get
\begin{align}\label{sketch_f1}
&\frac{1}{2} \int \ball{i}{2}(x_3)\ball{i}{2}(x_1)   \frac{G(x_1,x_2,x_3;\bz)}{x_3-x_1} (F_1(x_1,x_2) - F_1(x_3,x_2)) \prod_{j=1}^3 d^2 x_j\,.
\end{align}
Next we want to argue that when $x_1,x_3 \in B_2$, the factor $F_1(x_1,x_2) - F_1(x_3,x_2)$ behaves like $(x_3-x_1)$ multiplied by something that is integrable against $G$. Indeed, by using $\mathbf{1}_{B_2}(x_j) \mathbf{1}_{B_1}(x_j) = \mathbf{1}_{B_2}(x_j)$ we get
\begin{align*}
&\ball{i}{2}(x_3)\ball{i}{2}(x_1)(F_1(x_1,x_2) - F_1(x_3,x_2)) \\
 &= \ball{i}{2}(x_3)\ball{i}{2}(x_1) \left( \frac{\ball{i}{1}(x_1) \ballc{i}{1}(x_2)}{x_1-x_2} - \frac{\ball{i}{1}(x_3)  \ballc{i}{1}(x_2)}{x_3-x_2} \right) \\
&= \frac{\ball{i}{2}(x_1) \ballc{i}{1}(x_2)}{x_1-x_2} \ball{i}{2}(x_3) - \frac{\ball{i}{2}(x_3) \ballc{i}{1}(x_2)}{x_3-x_2} \ball{i}{2}(x_1)\,,
\end{align*}
and both of these factors are smooth (and bounded) when $x_1,x_3 \in B_2$. This fact is quite intuitive since the singular behaviour of the function $F_1(x,y)$ happens on the circle $\partial B_1$ and thus the singularity is gone when $x$ is restricted to the smaller ball $B_2$. From this we infer that
\begin{align*}
\ball{i}{2}(x_3)\ball{i}{2}(x_1)(F_1(x_1,x_2) - F_1(x_3,x_2)) &= \mathcal{O}(x_3-x_1) H(x_2)\,,
\end{align*}
where $H$ is bounded. When we insert this to (\ref{sketch_f1}) we get a nice integrable term. The $\ballc{i}{2}(x_1)$-term is
\begin{align*}
\int \ball{i}{2}(x_3)\ballc{i}{2}(x_1)   \frac{G(x_1,x_2,x_3;\bz)}{x_3-x_1} F_1(x_1,x_2) \prod_{j=1}^3 d^2 x_j &  = \int F_2(x_3,x_1) F_1 (x_1,x_2) G(x_1,x_2,x_3;\bz) \prod_{j=1}^2 d^2 x_j\,.
\end{align*}
Terms like this are shown to be integrable by using the estimate
\begin{align}\label{sketch_f2}
|F_j(x,y)| & \leq |F_j(x,y)| \ba{i}{j}(x)\bac{i}{j}(y) +   C\,,
\end{align}
because after inserting this, we get terms where one variable has singularity only on one of the circles $\partial B_j$ (by disjointness of the annuli $A_j$) and since the radius $j$ is different for each $F$-factor, none of these singularities "stack". This is proven in Proposition \ref{integrability}.

Next we integrate by parts the left-hand side of (\ref{sketch_f0}). We get

\begin{align*}
\int \ball{i}{2}(x_3) F_1(x_1,x_2) \partial_{x_3} G(x_1,x_2,x_3;\bz) \prod_{j=1}^3 d^2 x_j = - \oint_{\partial B_2} d x_3  \int F_1(x_1,x_2) G(x_1,x_2,x_3;\bz) \prod_{j=1}^2 d^2 x_j\,.
\end{align*}
To show that this is convergent we again split the integral into the part $\ba{i}{1}(x_1) \bac{i}{1}(x_2)$ and its complement. In the complement $F$ is bounded so the integral is clearly absolutely integrable. For the other part our fusion estimate \ref{fusion} implies
\begin{align*}
F_1(x_1,x_2) G(x_1,x_2,x_3;\bz) \ba{i}{1}(x_1) \bac{i}{1}(x_2) \ball{i}{2}(x_3)  \leq C \frac{\ba{i}{1}(x_1) \bac{i}{1}(x_2)}{|x_1-x_2|^{3-\zeta}}\,,
\end{align*}
where $C$ depends on $\delta = \min_{i,j=1; i \neq j}^N |z_i-z_j|$ and $\zeta > 0$. This is integrable.

\subsection{Higher derivatives}

Denote $\bx = (x_1,\hdots,x_{2n})$. When we start to compute higher order derivatives, integrals of the form
\begin{align*}
\int F(\bx) G(\bx;\bz) \prod_{j=1}^n d^2 x_j\,,
\end{align*}
where $F$ is a product of functions of the form \eqref{sketch_F}, will start to appear. For each new derivative $\partial_{z_i}$ we add new $\gamma$-insertions (insertions with Liouville momentum $\alpha = \gamma$) to $G$ and a factor $F_j(x_a,x_b)$, with some indices $a \neq b$, to $F$, where the index $j$ tells the radius of the ball appearing in the definition of $F_j$ (we increment $j$ after each differentiation so all the balls have different radii). The most singular term appearing in the formula for $\partial_{z_i}^n G(\bz)$ will then be 
\begin{align*}
\int \prod_{j=1}^n F_j (x_j,y_j) G(\bx,\by;\bz)  d^2 x_j \, d^2 y_j\,,
\end{align*}
where $\bx = (x_1,\hdots,x_n)$ and $\by = (y_1,\hdots,y_n)$. The proof of convergence of this integral is essentially the same as in the $C^2$-case we did above: just use the estimate (\ref{sketch_f2}) and the $n$-pair fusion estimate \ref{fusion}. Then to show differentiability of this integral in the $z_i$'s, we have to take the derivative of $G$ in the above integral and this leads us to investigate the integrals
\begin{align}\label{sketch_n}
\int \prod_{j=1}^n F_j (x_j,y_j) \frac{G(\bx,x_{n+1},\by;\bz)}{x_k-z_i} d^2 x_j \, d^2 y_j \, d^2 x_{n+1}\,,
\end{align}
where $k \in \{1,\hdots,n+1\}$. This we simplify by using the same integration by parts argument as in the $C^2$-case. Thus we write the integral
\begin{align*}
\int \ball{}{n+1}(x_k) \prod_{j=1}^n F_j(x_j,y_j) \partial_{x_k} G(\bx,x_{n+1},\by;\bz) d^2 x_j \, d^2 y_j \, d^2 x_{n+1}
\end{align*}
in two different ways (by using the derivative formula $\partial_{x_k} G$ and by integration by parts). Then we can solve for (\ref{sketch_n}). In the resulting expression we have to simplify the integrals
\begin{align*}
\int \ball{i}{n+1}(x_k) \prod_{j=1}^n F_j (x_j,y_j) \frac{G(\bx,x_{n+1},\by;\bz)}{x_k-x_l} d^2 x_j \, d^2 y_j \, d^2 x_{n+1}\,,
\end{align*}
where $l \neq k$. Again, we insert $1 = \ball{i}{n+1}(x_l) + \ballc{i}{n+1}(x_l)$. The latter part produces the integral
\begin{align*}
\int F_{n+1}(x_k,x_l) \prod_{j=1}^n F_j (x_j,y_j) G(\bx;x_{n+1},\by;\bz)  d^2 x_j \, d^2 y_j \, d^2 x_{n+1}\,,
\end{align*}
which converges by the fusion estimate. The remaining part we symmetrize (as before). Note that the factors $F_j$ that depend on $x_k$ are smooth and bounded in $B_{n+1}$ and the same holds for $x_l$. Thus after symmetrizing, the parts which depend on these variables produce a $\mathcal{O}(x_k-x_l)$-term and the parts that do not depend on these variables remain the same. Thus we get
\begin{align*}
&\int \ball{i}{n+1}(x_k) \ball{i}{n+1}(x_l) \prod_{j=1}^n F_j (x_j,y_j) \frac{G(\bx,x_{n+1},\by;\bz)}{x_k-x_l} d^2 x_j \, d^2 y_j \, d^2 x_{n+1} \\
&= \frac{1}{2} \int \ball{i}{n+1}(x_k) \ball{i}{n+1}(x_l) \underset{j \notin \{k,l\}}{\prod_{j=1}^n} F_j (x_j,y_j) \varphi(\bx,x_{n+1},\by ) G(\bx,x_{n+1},\by;\bz)  d^2 x_j  d^2 y_j  d^2 x_{n+1}
\end{align*}
where $\varphi$ is bounded. Integrability of this follows from the fusion estimate.

\section{Proof of smoothness}\label{section_proof}

In this section we give the detailed proof of Theorem \ref{theorem}. We want to iterate the computation we did in the $C^1$ proof. What we will observe later is that the derivatives of $G(\bz)$ can be expressed as integrals of singular functions against $G$ with additional $\gamma$-insertions.

Let $N \geq 3$ denote the amount of insertion points in the correlation function and $n \in \Z_+$ the amount of partial derivatives we are taking. Fix a sequence of numbers $i(j) \in \{1,\hdots N\}$ where $j \in \{1,\hdots,n\}$. We study the $n$th order partial derivative
\begin{align*}
\partial_{z_{i(1)}} \hdots \partial_{z_{i(n)}} G(\bz)\,,
\end{align*}
where $\bz = (z_1,\hdots,z_N) \in U_N$. Let $B_j = B(z_{i(j)},r/j)$, $A_j$ be the corresponding closed annulus containing $\partial B_j$ so small that $(A_j)_{j=1}^n$ are disjoint (see Figure \ref{figure}, which describes the notation in the case of the partial derivative $\partial_{z_1} \partial_{z_2} \partial_{z_1} G(z_1,z_2,z_3,z_4)$), and 
\begin{align}
F_j(x,y) = \frac{\ball{}{j}(x) \ballc{}{j}(y)}{x-y}\,.
\end{align}
To set up a suitable induction, we define the following function classes.
\begin{definition}\label{class}
By $\sF_n$ we denote the set of functions which are linear combinations of functions of the form
\begin{align*}
\varphi(\bx) \prod_{j \in J} F^{\ij}_j (x_\kj,  x_{ \tkj})
\end{align*}
where $\bx = (x_1,\hdots,x_{2n})$, $J \subset \{1,\hdots,n\}$ and
\begin{enumerate}
\item $1 \leq \kj \leq 2n$ for all $j \in J$
\item $ \tkj \in \{1,\hdots,2n\} \setminus \{\kj\}$
\item $\varphi$ is bounded on $\C^{2n}$
\item $x_k \mapsto \varphi(\bx)$ is $C^\infty$ outside of the circles $\partial B_j$, $j \in \{1,\hdots,n\}$.
\end{enumerate}
\end{definition}
We will need the following two properties of this function class.
\begin{lemma}\label{smooth_lemma}
Let $F \in \sF_n$. Then $x_k \mapsto F(\bx) \ball{}{n+1}(x_k)$ is $C^\infty$ in $B_{n+1}$.
\end{lemma}
\begin{proof}
Follows from the definition of $F_j$ and property 4 in Definition \ref{class}.
\end{proof}

\subsection{Simplification}
\begin{remark}
Strictly speaking we should do the following computations with the regularized functions $G_\eps$ and $\frac{1}{(x_a-x_b)_{\eps,\eps}}$, but this does not affect any of the algebraic manipulations we do, and interchanging limits and integrals works in the end easily by our integrability result (Proposition \ref{integrability}) and the estimates for $\sup_{\eps > 0} G_\eps(x,y;\bz)$ in \cite{ward}. Thus we choose not to write all the epsilons in our computations.
\end{remark}

\begin{lemma}\label{symmetrization}
Let $F \in \sF_n$, $a \in \{1,\hdots,2n+2 \}$ and $b \in \{1,\hdots,2n+2\} \setminus \{a\}$. Then
\begin{align*}
\int \ball{}{n+1}(x_a) F(\bx) \frac{G(\bx,x_{2n+1},x_{2n+2};\bz)}{x_a-x_b}  \prod_{j=1}^{2n+2}d^2 x_j &= \int \tilde F(\bx,x_{2n+1},x_{2n+2}) G(\bx,x_{2n+1},x_{2n+2};\bz) \prod_{j=1}^{2n+2} d^2 x_j \,,
\end{align*}
where $\tilde F \in \sF_{n+1}$.
\end{lemma}
\begin{proof}
We insert $1 = \ball{}{n+1}(x_b) + \ballc{}{n+1}(x_b)$ into the integral. We symmetrize the first term to get
\begin{align*}
\frac{1}{2} \int \ball{}{n+1}(x_a) \ball{}{n+1}(x_b) \frac{G(\bx,x_{2n+1}, x_{2n+2};\bz)}{x_a- x_b} (F(\bx) - F(\bx;x_a \leftrightarrow x_b))\,,
\end{align*}
where 
\begin{align*}
F(\bx;x_a \leftrightarrow x_b) &:= F(x_1,\hdots,x_{a-1},x_b,x_{a+1},\hdots,x_{b-1},x_a,x_{b+1},\hdots,x_{2n})\,,
\end{align*}
when $x_a, x_b \in \{x_1,\hdots,x_{2n}\}$, and
\begin{align*}
F(\bx;x_a \leftrightarrow x_b) &:= F(x_1,\hdots,x_{a-1},x_b,x_{a+1},\hdots,\hdots,x_{2n})\,,
\end{align*}
when $x_a \in \{x_1,\hdots,x_{2n}\}$ and $b \in \{x_{2n+1}, x_{2n+2}\}$, and finally
\begin{align*}
F(\bx;x_a \leftrightarrow x_b) &:= F(\bx)\,,
\end{align*}
when $x_a, x_b \in \{ x_{2n+1}, x_{2n+2}\}$. By definition of $\sF_{n+1}$ the function $F$ is smooth in $x_a$ and $x_b$ in the integration domain and thus the difference $F(\bx) - F(\bx;x_a \leftrightarrow x_b)$ can be written as $(x_a-x_b)\varphi(\bx,x_{2n+1}, x_{2n+2}) H(\bx)$ where $H \in \sF_{n+1}$ is not a function of $x_a$ and $x_b$, and $\varphi$ is bounded everywhere and smooth outside of the circles $\partial B_j$. Indeed, the function $H$ will consists of the factors $F_j$ in $F$ that are not functions of $x_a$ and $x_b$. The function $\ball{}{n+1}(x_a) \ball{}{n+1}(x_b) \varphi(\bx,x_{2n+1}, x_{2n+2}) H(\bx)$ belongs to $\sF_{n+1}$.

The $\ballc{}{n+1}(x_b)$-part of the integral immediately becomes
\begin{align*}
\int F_{n+1}(x_a,x_b) F(\bx) G(\bx,x_{2n+1}, x_{2n+2})  \prod_{j=1}^{2n+1} d^2 x_j \,,
\end{align*}
and clearly the function $F_{n+1}(x_a,x_b) F(\bx)$ belongs to $\sF_{n+1}$.
\end{proof}

Next we show that the integrals that pop up when we compute derivatives of $G$ can be expressed in terms of absolutely convergent integrals. 
\begin{proposition}\label{prop2}
Let $F \in \sF_n$, $k \in \{1,\hdots,2n+1\}$ and $i \in \{1,\hdots,N\}$. Then 
\begin{align*}
\int F(\bx) \frac{G(\bx,x_{2n+1};\bz)}{x_k-z_i} \prod_{j=1}^{2n+1} d^2 x_j &= \int \tilde F(\bx,x_{2n+1},x_{2n+2}) G(\bx,x_{2n+1},x_{2n+2};\bz) \prod_{j=1}^{2n+1} d^2 x_j \\
& \quad + \oint_{\partial B_{n+1}} \, d  x_k \, \int F(\bx) G(\bx,x_{2n+1};\bz) \prod_{j=1; j \neq k}^{2n+1} d^2 x_j  \,,
\end{align*}
where $\tilde F \in \sF_{n+1}$.
\end{proposition}

\begin{proof}
\begin{align*}
&\int \ball{}{n+1} (x_k) F(\bx) \partial_{x_k} G(\bx,x_{2n+1};\bz) \prod_{j=1}^{2n+1} d^2 x_j \\
&=  -\int \ball{}{n+1}(x_k) F(\bx) G(\bx,x_{2n+1};\bz) \sum_{j=1}^N \frac{\alpha_j \gamma}{2} \frac{1}{x_k-z_j} \prod_{j=1}^{2n+1} d^2 x_j \\
& -\int \ball{}{n+1}(x_k) F(\bx) G(\bx,x_{2n+1};\bz) \sum_{j=1;j \neq k}^{2n+1} \frac{\gamma^2}{2} \frac{1}{x_k-x_j} \prod_{j=1}^{2n+1} d^2 x_j \\
& \quad + \frac{\mu \gamma^2}{2} \int \ball{}{n+1}(x_k) F(\bx) \frac{G(\bx,x_{2n+1},x_{2n+2};\bz)}{x_k-x_{2n+2}} \prod_{j=1}^{2n+2} d^2 x_j\,.
\end{align*}
We want to solve for the $j=i$ term in the first sum. The $j \neq i$ terms need no simplification. The rest of the terms simplify correctly by Lemma \ref{symmetrization}.

When we integrate by parts the left-hand side we get the terms
\begin{align*}
&- \oint_{\partial B_{n+1}} \int F(\bx) G(\bx,x_{2n+1};\bz) \prod_{j=1;j \neq k}^{2n+1} d^2 x_j \, d  x_k - \int \ball{}{n+1}(x_k)  \partial_{x_k} F(\bx) G(\bx,x_{2n+1};\bz) \prod_{j=1}^{2n+1} d^2 x_j \,.
\end{align*}
When $x_k \in B_{n+1}$, $\partial_{x_k} F(\bx)$ can be written as $\varphi(\bx) H(\bx)$ where $H \in \sF_{n+1}$ does not depend on $x_k$ and $\varphi$ is bounded everywhere and smooth outside the circles $\partial B_j$. Thus $\ball{}{n+1}(x_k) \partial_{x_k} F(\bx) \in \sF_{n+1}$.

The boundary integral is convergent because 
\begin{align*}
\sup_{x_k \in \partial B_{n+1}} \left| \int F(\bx) G(\bx,x_{2n+1};\bz)  \prod_{j=1;j \neq k}^{2n+1} d^2 x_j \right| & = C(\delta) < \infty
\end{align*}
by Lemma \ref{uniform}.

\end{proof}

\subsection{Integrability}

Next we show that the integrals appearing in Proposition \ref{prop2} are absolutely convergent.

\begin{proposition}\label{integrability}
Let $F \in \sF_n$. Then
\begin{align*}
\int |F(\bx) G(\bx;\bz) | \dx{2n} < \infty\,.
\end{align*}
\end{proposition}
\begin{proof}
Let $a(j)$, $b(j)$ be as in Definition \ref{class}. We may assume that $F(\bx) = \varphi(\bx) \prod_{j \in J} F^{\ij}_{j}(x_\kj, x_\tkj)$ where $J \subset \{1,\hdots,n\}$. Since $\varphi$ is globally bounded, we may also assume that $\varphi \equiv 1$. Next, we can apply KPZ-formula of Lemma \ref{kpz} to obtain
\begin{align*}
\int | \prod_{j \in J} F^\ij_j(x_\kj, x_\tkj)| G(\bx;\bz) | \dx{2n} &= C \int | \prod_{j \in J} F_j(x_\kj, x_\tkj)| G(\{x_\kj, x_\tkj \}_{j \in J}; \bz) \dxj{J} \,,
\end{align*}
i.e. we integrate away the $\gamma$-insertions which do not appear in the function $F$. We split the integrals and estimate as follows
\begin{align*}
|F^\ij_j(x_\kj,x_\tkj)| &= |F^\ij_j(x_\kj,x_\tkj)| (\ann{\ij}{j}(x_\kj) + \annc{\ij}{j}(x_{\kj}))(\ann{\ij}{j}(x_\tkj) + \annc{\ij}{j}(x_\tkj)) \\
& \leq |F^\ij_j(x_\kj,x_\tkj)| \ann{\ij}{j}(x_\kj) \ann{\ij}{j}(x_\tkj) + C \,.
\end{align*}
Thus 
\begin{align*}
& \int | \prod_{j \in J} F^\ij_j(x_\kj, x_\tkj)| G(\{x_\kj, x_\tkj \}_{j \in J}; \bz)d^2 x_\kj \, d^2 x_\tkj \\
& \leq C  \sum_{J' \subset J} \int \prod_{j \in J'} |F^\ij_j(x_\kj, x_\tkj)|  \ann{\ij}{j}(x_\kj) \ann{\ij}{j}(x_\tkj) G(\{x_\kj, x_\tkj \}_{j \in J'}; \bz) \, d^2 x_\kj \, d^2 x_\tkj \,.
\end{align*}
Each of the $x$-variables appear in only one of the $F^\ij_j$-factors since otherwise the integrand vanishes (because $A_j$ are disjoint). Thus we can use the $n$-pair fusion estimate  (Lemma \ref{fusion}) to get
\begin{align*}
& \leq C \sum_{J' \subset J} \int \prod_{j \in J'} |F^\ij_j(x_\kj, x_\tkj)|  \ann{\ij}{j}(x_\kj) \ann{\ij}{j}(x_\tkj)  |x_{\kj}-x_\tkj|^{-2+\zeta} \, d^2 x_\kj \, d^2 x_\tkj \\
&= C \sum_{J' \subset J} \int \prod_{j \in J'}  \ba{\ij}{j}(x_\kj) \bac{\ij}{j}(x_\tkj)   |x_{\kj}-x_\tkj|^{-3+\zeta} \, d^2 x_\kj \, d^2 x_\tkj\,.
\end{align*}
This converges by Lemma \ref{integral}.
\end{proof}
\subsection{Boundary terms}\label{boundary_terms}

We still have to show that the boundary term appearing in Proposition \ref{prop2} is integrable and that the boundary terms are differentiable and satisfy the analogues of Propositions \ref{prop2} and \ref{integrability}.

\begin{lemma}\label{uniform}
Let $J \subset \{1,\hdots,n\}$, $L \subset \{1,\hdots,2n\}$ and  $(K_j)_{j \in J^c}$ be compact subsets of $\C$ that are disjoint from each other and the circles $\partial B(z_{i(j)},r/j)$. For each $j \in J$ fix numbers $a(j),b(j) \in \{1,\hdots,2n\}$, $a(j) \neq b(j)$. Denote $\bx = (x_1,\hdots,x_{2n})$ and let $\bx \mapsto \varphi(\bx)$ be bounded. Then 
\begin{align*}
\sup_{i \in L^c} \sup_{x_i \in K_i} \int \prod_{l \in L} |\varphi(\bx)| \prod_{j \in J} |F_j(x_\kj, x_\tkj)| G(\bx;\bz)  d^2 x_l  < \infty \,.
\end{align*}
\end{lemma}

\begin{remark}
The role of this lemma is to show that the boundary integrals in \eqref{proof_simplified_derivative} converge. In this case the boundary integrals are over circles which are examples of the sets $K_j$ appearing in the statement of the lemma.
\end{remark}

\begin{proof}
This follows from the fact that in the proof of Lemma \ref{fusion} the $\delta$-dependent constant $C_\delta$ satisfies $\sup_{\delta \in K} |C_\delta| < \infty$ whenever $K$ is a compact set disjoint from the origin. This is easy to see since by taking a bit more care of $C_\delta$ one sees that it can be chosen to be $C \delta^a$ for some constant $C$ and  $a <0 $.
\end{proof}

The above lemma says that the boundary integrals that we see are integrals of bounded functions over compact sets so they converge. In addition to this we need that the boundary integral terms appearing in our iteration are differentiable with respect to the insertions $z_i$. This follows from the following lemma.
\begin{lemma}\label{boundary_derivative}
Let $J \subset \{1,\hdots,n\}$, $L \subset \{1,\hdots,2n\}$ and $i \in \{1,\hdots,N\}$. Fix some indices $a(j),b(j) \in \{1,\hdots,2n\}$, $a(j) \neq b(j)$, for each $j \in J$.

Then for $k \in L \cup \{2n+1\}$ we have
\begin{align*}
&\prod_{l_2 \in L^c} \mathbf{1}_{\partial B_{l_2}}(x_{l_2}) \prod_{l_1 \in L} \int \prod_{j \in J} F_{j}(x_{a(j)}, x_{b(j)}) \frac{G(\bx,x_{2n+1};\bz)}{x_k - z_i} \, d^2 x_{l_1} \, d^2 x_{2n+1}\\
&= \prod_{l_2 \in L^c} \mathbf{1}_{\partial B_{l_2}}(x_{l_2}) \prod_{l_1 \in L} \int  \tilde F(\bx,x_{2n+1},x_{2n+2})  G(\bx,x_{2n+1},x_{2n+2};\bz) \, d^2 x_{l_1} \, d^2 x_{2n+1} \, d^2 x_{2n+2} \\
& \quad + \prod_{l_2 \in L^c} \mathbf{1}_{\partial B_{l_2}}(x_{l_2}) \oint_{\partial B_{n+1}} \, d  x_k\, \prod_{l_1 \in L \setminus \{k\}} \int \prod_{j \in J \setminus \{k\}} F_{j}(x_{a(j}, x_{b(j)})  G(\bx,x_{2n+1};\bz) \, d^2 x_{l_1} \, d^2 x_{2n+1}\,,
\end{align*}
where $\tilde F \in \sF_{n+1}$.

For $k \in L^c$ we have the same formula without the boundary integral term since 
\begin{align*}
\prod_{j \in J} F_{j}(x_{a(j)}, x_{b(j)}) \frac{1}{x_k - z_i} \in \sF_n\,.
\end{align*}
\end{lemma}
\begin{proof}
The case $k \in L \cup \{2n+1\}$ is exactly the same as for Proposition \ref{prop2} and the case $k \in L^c$ follows from the definition of $\sF_n$.
\end{proof}

\subsection{Proof of Theorem \ref{theorem}}
By combining the derivative formulas (\ref{derivative}) and (\ref{c1}) with Proposition \ref{prop2} and the corresponding results for the boundary terms in Section \ref{boundary_terms}, we see that 
\begin{align}\label{proof_simplified_derivative}
&\partial_{z_{i(1)}} \hdots \partial_{z_{i(n)}} G(\bz) \nonumber \\
& \quad = \sum_{k=1}^{2n} \sum_{J \subset \{1,\hdots,2k\}} C_{J,k} \prod_{j_1 \in J} \prod_{j_2 \in J^c} \oint_{\partial B_{j_2}} d  x_{j_2} \int_\C d^2 x_{j_1} F_{J,k}(x_1,\hdots,x_{2k}) G(\bx;\bz)\,,
\end{align}
where $\bx = (x_1,\hdots,x_{2k})$, $C_{J,k}$ are some constants and $F_{J,k} \in \sF_k$ is a linear combination of functions of the form
\begin{align*}
F_{J,k}(\bx) &= \varphi(\bx) \prod_{j \in J} F_j(x_\kj, b_\tkj)\,,
\end{align*} 
where $\varphi$ is bounded and $a(j) \neq b(j)$ are some arbitrary choice of indices. Combining this with Proposition \ref{integrability} together with Section \ref{boundary_terms} we see that all these integrals are absolutely convergent.

Taking $\partial_{\bar z_i}$-derivatives works the same way since in the derivation of the derivative formula (\ref{derivative}) one uses Gaussian integration by parts, which leads to terms containing derivatives of the form $\partial_x C_g(x,y)$ of the correlation of the GFF. When computing $\partial_{\bar z_i}$ instead of $\partial_{z_i}$ the derivative $\partial_x C_g(x,y)$ gets replaced by $\partial_{\bar x} C_g(x,y)$ and the essential term in $C_g(x,y)$ is $\ln \frac{1}{|x-y|}$ which is symmetric in $x-y$ and $\bar x - \bar y$. So in the end everything works the same way in the $\partial_{\bar z_i}$ case.

We have established the smoothness in the case of the round metric (\ref{round_metric}). The generalization for any diagonal metric follows from the Weyl anomaly (Theorem 3.11 in \cite{DKRV})
\begin{align*}
\langle \prod_{i=1}^N V_{\alpha_i}(z_i) \rangle_{e^\varphi g} &= e^{A(\varphi)} \langle \prod_{i=1}^N V_{\alpha_i}(z_i) \rangle_g\,,
\end{align*}
where
\begin{align*}
e^{A(\varphi)} &= \exp \left( \frac{c_L-1}{24 \pi }  \int ( |\partial_z \varphi(z)|^2  + \tfrac{1}{2} g(z) R_g(z) \varphi(z) ) d^2 z  \right)\,. 
\end{align*}
Now clearly if we have smoothness in the metric $g$, then we get smoothness for any metric $e^\varphi g$ in the same conformal class, that is, for any diagonal metric.

\begin{flushright}
$\qed$
\end{flushright}

\appendix

\section{Sobolev spaces on the Riemann sphere}\label{sobolev_appendix}

Let $g$ be a Riemannian metric on the Riemann Sphere $\rs$. The associated Laplace-Beltrami operator is given by
\begin{align*}
\Delta_g &:= \frac{1}{\sqrt{\det g}} \sum_{i,j=1}^2 \partial_i ( \sqrt{\det g} g^{ij} \partial_j)\,,
\end{align*}
where $(g^{ij})_{i,j=1}^2$ are the components of the inverse of $g$. The operator $-\Delta_g$ has eigenfunctions $(\varphi_i)_{i=0}^\infty$ with non-negative eigenvalues $(\lambda_i)_{i=0}^\infty$. The eigenvalues satisfy
\begin{align*}
0 = \lambda_0 < \lambda_1 < \lambda_2 < \hdots
\end{align*}
and the eigenfunctions $(\varphi_i)_{i=0}^\infty$ form an orthonormal basis of $L^2(\rs,g)$.

We define the Sobolev spaces $\ssob$ on the Riemann sphere by
\begin{align}\label{sobolev}
\ssob &:= \left \{ f = \sum_{i=0}^\infty f_i \varphi_i : (f_i)_{i=0}^\infty \in \R^\N, \; \|f\|_\ssob := \sum_{i=1}^\infty |f_i|^2 \lambda_i^s  < \infty \right\}\,.
\end{align}
We denote the subspace of zero mean elements of $\sob$ by $\sobo$
\begin{align*}
\ssobo &:= \{ f \in \ssob :  \int_\C f(z) \,  \opn{vol}_g(d^2z) = 0 \}\,,
\end{align*}
where $\opn{vol}_g$ is the volume form of $g$. For $s<0$ this means the elements satisfying $\langle f,1\rangle = 0$ where $\langle \cdot , \cdot \rangle$ is the dual bracket. From (\ref{sobolev}) we see that this subspace corresponds to the elements satisfying $f_0=0$. From the sequence representation it is easy to see that the continuous dual of $\ssob$ is $\ssobd$ and the continuous dual of $\ssobo$ is $\ssobod$. The zero mean spaces become Hilbert spaces when endowed with the inner product
\begin{align*}
\langle f,h \rangle_\ssobo &:= \sum_{i=1}^\infty f_i h_i \lambda_i^s\,.
\end{align*}
Note that in the case $s=1$ the Sobolev norm agrees with the Dirichlet energy
\begin{align*}
\langle f,h \rangle_\sobo^2 &=  \int_\C \nabla_g f(z) \cdot \nabla_g h(z) \, \opn{vol}_g(d^2z) z\,, \quad \|f\|_\sobo^2 = \langle f,f, \rangle_\sobo\,,
\end{align*}
where $\nabla_g$ is the $g$-gradient. The covariance of the zero mean GFF satisfies
\begin{align*}
\E[X(f)X(h)] &= \int_{\C^2} f(z) h(w) C_g(z,w) \,\vol{z} \vol{w}= \langle f,h \rangle_{\dsobo}
\end{align*}
for any $f,h \in \sobo$.

\section{Lemma for the fusion estimate}\label{fusion_appendix}

In this section we work with the GFF with zero mean over the unit circle, which we denote by $X_0$. It has the covariance
\begin{align}\label{new_covariance}
C_0(x,y) &= \ln \frac{1}{|x-y|} + \mathbf{1}_{\{|x| \geq 1\}} \ln |x|  + \mathbf{1}_{\{|y| \geq 1\}} \ln |y| \,.
\end{align}
Changing the zero mean GFF to the zero circle average GFF corresponds to shifting the constant $c$ in $c+X$. Indeed if $X$ is the GFF with zero mean over the whole complex plane, then
\begin{align*}
c + X \overset{law}{=} c + X -  \int_0^{2\pi} X(e^{i \theta}) \frac{d \theta}{2\pi}\,.
\end{align*}
The term $X - \int_0^{2\pi} X(e^{i\theta}) \frac{d \theta}{2\pi}$ can be identified as the zero circle average GFF $X_0$.

We use the radial decomposition of the GFF
\begin{align*}
X_0(x) &= X^r(|x|) + Y(x)\,,
\end{align*}
where $t \mapsto X^r(e^{-t})-X^r(1)$ is the Brownian motion and $Y$ is a Gaussian process called the lateral noise, see \cite{seiberg}. Plugging this into the chaos measure $M_\gamma^0$ of $X_0$ we get
\begin{align*}
M^0_\gamma(d^2x) &= c_\gamma g(x) |x|^{\frac{\gamma^2}{2}} e^{\gamma X^r(|x|)} M_\gamma(d^2x, Y)\,,
\end{align*}
where $M_\gamma(d^2 x, Y)$ is the GMC measure of the Gaussian field $Y$. Inside the unit disk (the purely log-correlated region of $X_0$) integrals of the GMC measure can now be written as
\begin{align}\label{radial_decomposition_integral}
\int_{\D} f(x) M^0_\gamma(d^2 x) &= c_\gamma \int_0^\infty \int_0^{2\pi} f(e^{-s} e^{i \sigma}) e^{\gamma B_s - \gamma Q s} g(e^{-s}) \mu_Y(ds,d\sigma)\,, 
\end{align}
where $B_s = X^r(e^{-s})$ is a Brownian motion and $\mu_Y$ is independent of $(B_s)_{s \geq 0}$.

Recall the formula
\begin{align}\label{appendix_correlation}
G(\bz) &= 2 B(\alpha) \mu^{-q} \gamma^{-1} \Gamma(q) \prod_{i<j} \frac{1}{|z_i-z_j|^{\alpha_i \alpha_j}} \E \left[ \left( \int_\C \mathcal{F} (x,\bz) M^0_\gamma(d^2 x) \right)^{-q}  \right]\,,
\end{align}
where $q = \frac{\sum_{i=1}^N \alpha_i - 2Q}{\gamma}$, $B(\alpha)$ is a constant depending on $\gamma$ and the $\alpha_i$'s and
\begin{align*}
\mathcal{F} (x,\bz) &= \prod_{i=1}^N \left( \frac{g(x)^{-\frac{1}{4}}}{|x-z_i|} \right)^{\gamma \alpha_i}\,.
\end{align*}
We want to derive a fusion estimate for $G(\bx,\by;\bz)$ in the case when $x_j$ merges with $y_j$ while the pair $(x_j,y_j)$ stays away from all the other insertions. Using (\ref{appendix_correlation}) we get 
\begin{align}\label{appendix_estimate}
G_\eps(\bx,\by;\bz) \prod_{j=1}^n \ba{i(j)}{j}(x_j) \bac{i(j)}{j}(y_j) & \leq C(\delta, \alpha, g) \prod_{j=1}^n |x_j-y_j|_\eps^{-\gamma^2} I_\eps\,, 
\end{align}
where $A_j$ and $B_j$ are as in the beginning of Section \ref{section_proof} and
\begin{align*}
I_\eps &= \E \left[ \left( \sum_{j=1}^n \int_{D_j}  \frac{1}{ |x-x_j|_\eps^{ \gamma^2} |x-y_j|_\eps^{\gamma^2}} M^0_{\gamma,\eps}(d^2 x) \right)^{-q} \right]\,,
\end{align*}
where $q = \frac{2 n \gamma + \sum_{i=1}^N \alpha_i - 2Q}{\gamma}$ and $D_j$ is an annulus with center at $x_j$, inner radius $|x_j-y_j|$ and outer radius $R_j$. We want to choose such $R_j$ that for $x \in D_j$ we have 
\begin{align*}
&\prod_{i=1}^N \left( \frac{g(x)^{-\frac{1}{4}}}{|x-z_i|_\eps} \right)^{\gamma \alpha_i} \prod_{i=1}^n \left( \frac{g(x)^{-\frac{1}{4}}}{|x-x_i|_\eps} \right)^{\gamma^2}  \prod_{i=1}^n \left( \frac{g(x)^{-\frac{1}{4}}}{|x-y_i|_\eps} \right)^{\gamma^2}  \\
& \quad \quad  \geq C(\bz,\alpha,g)  \frac{1}{|x-x_j|_\eps^{\gamma^2} |x-y_j|_\eps^{\gamma^2}}
\end{align*}
We can assume that $|x_j-y_j|$ is really small for every $j \in \{1,\hdots,n\}$. Then we choose any $R_j$ such that $D_j \subset A_j$ and the above estimate holds. We can also estimate
\begin{align*}
\int_\C  \mathcal{F}_\eps(x,\bz) M^0_{\gamma,\eps}(d^2x) & \geq \sum_{j=1}^n \int_{D_j} \mathcal{F}_\eps(x,\bz) M^0_{\gamma,\eps}(d^2 x)\,.
\end{align*}
By combining these two estimates we get (\ref{appendix_estimate}).

Next we derive another estimate that we will need. Inside $D_j$ we have  $|x-y_j|_\eps \leq C |x-x_j|_\eps$. Also, the correlation of the regularized field $X_{\eps}$ satisfies
\begin{align*}
\E[X_{0,\eps}(x)X_{0,\eps}(y)] & \leq C + \E[X_0(x)X_0(y)]\,,
\end{align*}
where $C$ is uniform in $\eps$, and thus we can use the Kahane Convexity Inequality \ref{kahane} to pass to the non-regularized measure $M^0_\gamma(d^2x)$. Without loss of generality we can assume that the points $x_j$ and $y_j$ fuse at the origin. By using the radial decomposition (\ref{radial_decomposition_integral}) of the GFF about the origin we get
\begin{align*}
&\E \left[ \left( \int_{D_j}  \frac{1}{|x-x_j|_\eps^{2\gamma^2}} M^0_{\gamma}(d^2 x) \right)^{-\frac{q}{n}} \right] = c_\gamma \E \left[ \left(  \int_0^{- \ln |x_j-y_j|_\eps} \int_0^{2\pi} e^{\gamma P_s} \mu_Y(ds, d \sigma) \right)^{-\frac{q}{n}}  \right]\,,
\end{align*}
where $P_s = B_s + (2\gamma -Q)s$. We split the integral by using the following events
\begin{align*}
M_{j,k} &= \{ \max_{s \in [0, - \ln |x_j-y_j|_\eps]} P_s \in [k-1,k] \}\,, \quad k \geq 1 \,, \\
M_{j,0} &= \{  \max_{s \in [0,- \ln |x_j-y_j|_\eps]} P_s \leq 0\}\,.
\end{align*}
We estimate the resulting integrals using the following lemma which is a special case of Lemma 6.5 in \cite{ward}.
\begin{lemma}\label{appendix_brownian_lemma}
Let $P_s = B_s + (2 \gamma-Q)s$ where $(B_s)_{s \geq 0}$ is a Brownian motion. Then for all $q>0$ we have
\begin{align*}
\E \left[ \frac{\mathbf{1}_{\{\sup _{u \in [0,r]} P_u \in [k-1,k]  \}}}{\left( \int_0^r \int_0^{2\pi} e^{\gamma P_s} \mu_Y (ds,d\sigma)  \right)^q} \right] & \leq C(k+1) e^{ (2 \gamma - Q - q \gamma)k} r^{- \frac{3}{2}} e^{- \frac{(2 \gamma-Q)^2}{2} r}\,. 
\end{align*}
\end{lemma}
Now we can estimate 
\begin{align*}
&\E \left[ \left(  \int_0^{- \ln |x_j-y_j|_\eps} \int_0^{2\pi} e^{\gamma P_s} \mu_Y(ds, d \sigma) \right)^{-\frac{q}{n}}  \right] \\
& \leq \sum_{k=0}^\infty \E \left[ \mathbf{1}_{M_k} \left(  \int_0^{-\ln|x_j-y_j|_\eps} \int_0^{2\pi} e^{\gamma P_s} \mu_Y(ds,d \sigma) \right)^{-\frac{q}{n}} \right] \\
& \leq C \sum_{k=0}^\infty (k+1)e^{(- Q - \frac{ \sum_{i=1}^N \alpha_i - 2Q}{n }  )k} |\ln |x_j-y_j|_\eps |^{- \frac{3}{2}} |x_j-y_j|_\eps^{ \frac{(2 \gamma-Q)^2}{2}}\,.
\end{align*}
The series converges since $Q + \frac{\sum_{i=1}^N \alpha_i -2Q}{n} > 0$. Thus we have shown that
\begin{align}\label{appendix_single_integral_estimate}
\E [W_j^{-\frac{q}{n}}] & \leq C |\ln |x_j-y_j|_\eps |^{3/2} |x_j-y_j|_\eps^{ \frac{(2 \gamma-Q)^2}{2}}\,,
\end{align}
where
\begin{align*}
W_j &= \int_{D_j}\frac{1}{ |x-x_j|^{ \gamma^2} |x-y_j|^{\gamma^2}} M^0_{\gamma}(d^2 x)\,.
\end{align*}

\end{document}